\documentclass[11pt,reqno]{amsart}
\allowdisplaybreaks[4]
\usepackage{graphicx} 
\usepackage{epsfig}
\usepackage{amssymb}
\usepackage{amsmath}
\usepackage{cite}
\newtheorem{theorem}{Theorem}

\newtheorem{proposition}[theorem]{Proposition}

\newtheorem{lemma}[theorem]{Lemma}
\newcommand{\be}{\begin{equation}}
\newcommand{\ee}{\end{equation}}
\newcommand{\bea}{\begin{eqnarray}}
\newcommand{\eea}{\end{eqnarray}}
\newcommand{\ba}{\begin{array}}
\newcommand{\ea}{\end{array}}
\newcommand{\bean}{\begin{eqnarray*}}
\newcommand{\eean}{\end{eqnarray*}}

\renewcommand{\L}{\mathcal{L}}

\newcommand{\M}{\mathcal{M}}
\newcommand{\Z}{\mathbb{Z}}

\renewcommand{\d}{\partial}
\newcommand{\N}{{\mathbb{N}}}

\newcommand{\pa}{\partial}

\renewcommand{\pa}{\partial}

\numberwithin{equation}{section}
\begin{document}

\title{Quantum torus symmetries of multicomponent modified KP hierarchy and reductions}
\author{Chuanzhong Li$^{1}$,\ Jipeng Cheng$^{2}$, }
\dedicatory {
$^1$ Department of Mathematics,  Ningbo University, Ningbo, 315211, China\\
$^2$ School of Mathematics, China University of
Mining and Technology, Xuzhou,   221116, China}
\thanks{Email addresses: lichuanzhong@nbu.edu.cn (C. Z. Li), chengjp@cumt.edu.cn (J. P. Cheng).}
\begin{abstract}
In this paper, we construct the multicomponent  modified KP hierarchy  and its  additional symmetries. The additional symmetries constitute an interesting multi-folds quantum torus type Lie algebra. By a reduction, we also construct the constrained multicomponent  modified KP hierarchy  and its Virasoro type additional symmetries.  \\
\textbf{Keywords}: multicomponent modified KP hierarchy,  quantum torus symmetries, additional symmetries, constrained multicomponent  modified KP hierarchy, Virasoro symmetry.\\
\textbf{PACS}: 02.30.Ik\\
\textbf{2010 MSC}: 35Q53, 37K10, 37K40
\end{abstract}
\maketitle

\tableofcontents
\section{Introduction}
The modified Kadomtsev-Petviashvili (mKP) hierarchy is one of the most important research objects
in the mathematical physics and integrable systems introduced in the early 1980s
\cite{djkm1,djkm2} and then  several versions \cite{chengy,dickey,takebe1,
takebe2,kacvdl,Zabrodin}, particularly the Kupershmidt-Kiso version \cite{kuper1,kiso} are developed. A Miura transformation will transfer it to the the KP systems like the relationship between KdV and mKdV equation.

Additional symmetries are a kind of important symmetries depending explicitly on the space and time variables\cite{OS}.
Recently, much important work about the KP type systems has been obtained in the additional symmetries, such as
the B and C type KP systems\cite{li1}.

As we know, the constrained KP hierarchy \cite{chengyiliyishenjmp1995} is another kind of KP type integrable hierarchy and the
additional symmetry flows  are not easily consistent with the
reduction on the  Lax operator for the constrained KP hierarchy \cite{aratyn1997}.
Therefore it is highly non-trivial  to find  a suitable  additional
symmetry  flows on eigenfunctions and adjoint eigenfunctions  to make the constraint compatible with its additional symmetry.
Besides, other KP type and Toda type systems  also have interesting structure of additional symmetries\cite{ourBlock,dtyptds,ghostdKP} such as  Block algebra \cite{Block}. The Block type additional structure was found in \cite{ourBlock} as a kind of
additional symmetry of the bigraded Toda hierarchy and later  a series of studies on integrable systems and Block algebras such as in \cite{dispBTH,dtyptds,blockDS,blockdDS} were done. After a quantization, the Block type Lie algebra becomes the so-called quantum torus Lie algebra \cite{takasakiquantum,torus,JNS}. Also a supersymmetric Block algebra was also found in supersymmetric BKP systems in \cite{NPB}.

The multi-component KP hierarchy is an important matrix-formed generalization of the original KP hierarchy and its additional symmetry was well-studied in \cite{dicheymdern}.
  Later with reduction, a quite natural
constrained multi-component KP hierarchy was studied by Y. Zhang in \cite{zhangJPA}.
 Recently the research on random matrices and non-intersecting Brownian motions and the study of moment matrices with regard to several weights showed that the determinants of such moment matrices these
determinants are tau-functions of the multi-component KP-hierarchy \cite{randommatrixCMP}.
 For multicomponent discrete integrable systems, the multicomponent Toda hierarchy recently attracts a lot of valuable research such as \cite{manas,EMTH,EZTH}.
  In \cite{TMP}, we considered the constrained multi-component KP hierarchy and identified its algebraic structure.
  In \cite{tian2011}, a q-deformed modified KP hierarchy and its additional
symmetries are constructed. Later in \cite{jipengmkp}, the additional symmetries of the mKP hierarchy and the constrained mKP hierarchy are constructed, and the corresponding actions
on two tau functions  are obtained.
For the multicomponent KP hierarchy, compact expressions for symmetry flows,  vertex operators and the associated Hamiltonian formalism was studied in \cite{orlov}. Also the relations to the isomonodromy problem and the related vector field action on Riemann surfaces of the system were later developed in \cite{orlov2}.
A natural question is what about the additional symmetries of the constrained multicomponent mKP hierarchy and other similar studies.

This paper will be organized in the following way. In Section 2, some
basic facts about the mKP hierarchy are introduced. The  multicomponent  modified KP hierarchy are constructed  in Section 3. Then in Section 4,
the additional symmetry of the multicomponent  modified KP hierarchy and the symmetry constitutes a quantum torus type Lie algebra. Next in Section 5, a constrained multicomponent  modified KP hierarchy are given. After that in Section 6, the  additional symmetry of the
constrained multicomponent  modified KP hierarchy are investigated.

\section{The modified KP hierarchy and multi-component modified KP hierarchy}
In this section, we will consider the algebra $g$ of the pseudo-differential operators
\begin{equation}\label{pdo}
  g=\{\sum_{i\ll \infty}u_i\pa^i\},
\end{equation}
where $\pa=\pa_x$ and $u_i=u_i(t_1=x,t_2,\cdots)$. For any operator $A=\sum_{i}a_i\pa^i\in g$,
$A_{\geq k}=\sum_{i\geq k}a_i\pa^i$ and $A_{<k}=\sum_{i<k}a_i\pa^i$.
A anti-involution operator $*$ satisfies: $(AB)^*=B^*A^*$,
$\pa^*=-\pa$, $f^*=f$ for any scalar function $f$.

The mKP hierarchy can be defined as the following Lax equation in Kupershmidt-Kiso version \cite{kuper1,kiso}
\begin{equation}\label{laxequation}
    \pa_{t_n} L=[(L^n)_{\geq 1},L],\ n=1,2,3,\cdots,
\end{equation}
with the Lax operator $L\in g$ as below
\begin{equation}\label{laxoperator}
    L=\pa+u_0+u_1\pa^{-1}+u_2\pa^{-2}+u_3\pa^{-3}+\cdots.
\end{equation}

The  mKP hierarchy
(\ref{laxequation}) contains the well-known mKP equation
\begin{equation}\label{mkpequation}
4u_{tx}=(u_{xxx}-6u^2u_{x})_x+3u_{yy}+6u_{x}u_{y}+6u_{xx}\int u_{y}dx.
\end{equation}

The Lax operator $L$ for
the mKP hierarchy can also be rewritten in terms of a dressing operator $Z$,
\begin{equation}\label{dressinglaxmdkp}
  L=Z\pa Z^{-1},
\end{equation}
where $Z$ is given by
\begin{equation}\label{dressingmdkp}
  Z=z_0+z_1\pa^{-1}+z_2\pa^{-2}+\cdots.
\end{equation}
Then the Lax equation (\ref{laxequation}) is equivalent to the following Sato equation
\begin{equation}\label{dressingmdkpequation}
  \partial_{t_n}Z=-L^n_{\leq 0}Z.
\end{equation}

In the next section, we will introduce an $N$-component modified KP hierarchy which contains $N$ infinite families of time variables $t_{\alpha, n}, \alpha=1,\ldots,N, n=1,2,\ldots$. The coefficients $A,u_1, u_2,\ldots$ of the Lax operator
\begin{equation}\label{Lax}L=A\pa +u_1\pa^{-1}+u_2\pa^{-2}+\ldots
\end{equation}
 are $N\times N$ matrices and $A=diag(a_1,a_2,\dots,a_N)$. There are another $N$ pseudodifferential operators $R_1,\ldots, R_N$ of the form
$$R_{\alpha}=E_{\alpha}+u_{\alpha, 1}\pa^{-1}+u_{\alpha, 2}\pa^{-2}+\ldots,$$ where $E_{\alpha}$ is the $N\times N$ matrix with $1$ on the $(\alpha,\alpha)$-component and zero elsewhere, and $u_{\alpha, 1}, u_{\alpha, 2},\ldots$ are also $N\times N$ matrices. The operators $L, R_1,\ldots, R_N$ satisfy the following conditions:
$$LR_{\alpha}=R_{\alpha} L, \quad R_{\alpha}R_{\beta}=\delta_{\alpha\beta}R_{\alpha}, \quad \sum_{\alpha=1}^N R_{\alpha}=E.$$ The Lax equations are:
\begin{equation*}
\frac{\pa L}{\pa t_{\alpha, n}}=[B_{\alpha, n}, L],\hspace{1cm}\frac{\pa R_{\beta}}{\pa t_{\alpha, n}}=[B_{\alpha, n}, R_{\beta}],\hspace{1cm}B_{\alpha, n}:=(L^n R_{\alpha})_{\geq 1}.
\end{equation*}
 The operator $\pa$ now is equal to $a_1^{-1}\pa_{t_{11}}+\ldots +a_N^{-1}\pa_{t_{N1}}$.
In fact the Lax operator $L$ and $R_{\alpha}$ can have the following dressing structures
\[L=PA\d P^{-1}, \ \ R_{\alpha}=PE_{\alpha}P^{-1},\]
where
\[P=E+\sum_{i=1}^{\infty}P_i\d^{-i},\] and the dressing operator $P$ satisfies the following Sato equations
\begin{equation*}
\frac{\pa P}{\pa t_{\alpha, n}}=-(L^n R_{\alpha})_{\leq 0}P.
\end{equation*}

Define the eigenfunction $w$ and the adjoint eigenfunction $w^*$ of the $N$-component  mKP hierarchy in the following way:
\begin{eqnarray}
w(t,\lambda)&=& P\left(e^{\xi(t,\lambda)}\right)=\hat w(t,\lambda)e^{\xi(t,\lambda)},\label{wavefunction}\\
w^*(t,\lambda)&=&(P^{-1}\pa^{-1})^*\left(e^{-\xi(t,\lambda)}\right)=\hat w^*(t,\lambda)\lambda^{-1}e^{-\xi(t,\lambda)},\label{adwavefunction}
\end{eqnarray}
with
\begin{eqnarray}
\xi(t,\lambda)&=&x\lambda+t_2\lambda^2+t_3\lambda^3+\cdots,\\
\hat{w}(t,\lambda)&=&z_0+z_1\lambda^{-1}+z_2\lambda^{-2}+\cdots,\label{hatwfunction}\\
\hat w^*(t,\lambda)&=&z_0^{-1}+z_1^*\lambda^{-1}+z_2^*\lambda^{-2}+\cdots,\label{hatwstarfunction}
\end{eqnarray}

and the eigenfunction $w$ and the adjoint eigenfunction $w^*$ of the $N$-component  mKP hierarchy satisfy
\begin{equation}\label{phipsieq}
 L(w(t,\lambda))=\lambda w(t,\lambda),\  L^*(w^*(t,\lambda))=\lambda w^*(t,\lambda),\ w_{t_{\alpha, n}}=(L^n R_{\alpha})_{\geq 1}w,\quad w^*_{t_{\alpha, n}}=-w^*\Big(\pa^{-1}(L^n R_{\alpha})_{\geq 1}^*\pa\Big).
\end{equation}
In the second equation here
it is assumed that the operators $\partial$ entering $\pa^{-1}(L^n R_{\alpha})_{\geq 1}^*\pa$ act to the left as $f\partial=f_x$.

\section{Additional symmetries of the multicomponent  mKP Hierarchy}
In order to define the additional symmetry for the multicomponent mKP hierarchy, the Orlov-Schulman operator  $M$ for the multicomponent mKP
hierarchy can be introduced as follows,
\begin{equation}\label{osoperatormkp}
    M=P\Gamma P^{-1},
\end{equation}
with $\Gamma =\sum_{\alpha=1}^N\sum_{i=1}^\infty it_{\alpha,i}E_{\alpha}\pa^{i-1}$.
From the fact $[\pa_{t_{\alpha, n}}-E_{\alpha}\pa^n,\Gamma]=0$, one can obtain
\begin{equation}\label{osmkpeq}
    [\pa_{t_{\alpha, n}}-(L^n R_{\alpha})_{\geq 1},M]=0,
    \quad i.e.\quad \pa_{t_{\alpha, n}} M=[(L^n R_{\alpha})_{\geq 1},M].
\end{equation}
Further
\begin{eqnarray}
L(w(t,\lambda))=\lambda w(t,\lambda),&&M(w(t,\lambda))=\pa_{\lambda}w(t,\lambda),
\end{eqnarray}
and
\begin{eqnarray}
[L,M]=1.
\end{eqnarray}

Then basing on a quantum parameter $q$, the additional flows for the time variable $t_{m,n,\alpha},t_{m,n,\alpha}^*$ are
defined respectively as follows
\begin{equation}
\partial_{t_{m,n,\alpha}} P=-(M^mL^nR_{\alpha})_{\leq 0}P,\  \partial_{t^*_{m,n,\alpha}}P =-(e^{mM}q^{nL}R_{\alpha})_{\leq 0}P,\ m,n \in \N,
\end{equation}
or equivalently rewritten as
\begin{equation}
\dfrac{\partial L}{\partial t_{m,n,\alpha}}=-[(M^mL^nR_{\alpha})_{\leq 0},L], \qquad
\dfrac{\partial M}{\partial t_{m,n,\alpha}}=-[(M^mL^nR_{\alpha})_{\leq 0},M],
\end{equation}

\begin{equation}
\dfrac{\partial L}{\partial t^*_{m,n,\alpha}}=-[(e^{mM}q^{nL}R_{\alpha})_{\leq 0},L], \qquad
\dfrac{\partial M}{\partial t^*_{m,n,\alpha}}=-[(e^{mM}q^{nL}R_{\alpha})_{\leq 0},M].
\end{equation}

Define the generator of the additional symmetries for multicomponent mKP hierarchy in the following double expansion
\begin{equation}\label{generatoraddsymmmkp}
    Y_{\alpha}(\lambda,\mu)=\sum_{m=0}^\infty\frac{(\mu-\lambda)^m}{m!}\sum_{l=-\infty}^\infty
    \lambda^{-l-m-1}(M^mL^{m+l}R_{\alpha})_{\leq 0}.
\end{equation}

Then by using the following fact

\begin{eqnarray}
{\rm res}_z\delta(\lambda,z)f(z)=f(\lambda),
\end{eqnarray}
one can
 find $Y_{\alpha}(\lambda,\mu)$ can be rewritten into a nice form similarly as \cite{dickeympla} (for KP system), which is listed in the proposition below.
\begin{proposition}\label{ylambdamuwpaw}
The generating operator $Y_{\alpha}(\lambda,\mu)$ have the following nice form as
\begin{eqnarray}
Y_{\alpha}(\lambda,\mu)&=&(R_{\alpha}w(t,\mu))\cdot\pa^{-1}\cdot w^*(t,\lambda)\pa.
\end{eqnarray}
And their combination $\sum_{\alpha=1}^N$ will produce the following generating operator
\begin{equation}\label{geneaddsymmmkpwave}
    Y(\lambda,\mu)=\sum_{\alpha=1}^NY_{\alpha}(\lambda,\mu)=w(t,\mu)\cdot\pa^{-1}\cdot w^*(t,\lambda)\pa.
\end{equation}
\end{proposition}
\begin{proof}
The proof is similar as the proof in \cite{jipengmkp,orlov} which is about the mKP hierarchy and the multicomponent KP hierarchy.
Firstly,  using \begin{eqnarray}
f\pa^{-1}=\sum_{i=1}^\infty\pa^{-i}f^{(i-1)}.
\end{eqnarray}
and
\begin{eqnarray}
{\rm res}_z\delta(\lambda,z)f(z)=f(\lambda),
\end{eqnarray} one can rewrite $(M^mL^{m+l}R_{\alpha})_{\leq0}$ into
\begin{eqnarray*}
&&(M^mL^{m+l}R_{\alpha})_{\leq0}={\rm res}_\pa(M^mL^{m+l}R_{\alpha}\pa^{-1})+\sum_{i=1}^\infty\pa^{-i}{\rm res}_\pa(\pa^{i-1}M^mL^{m+l}R_{\alpha})\\
&=& {\rm res}_\pa(M^mZ\pa^{m+l}E_{\alpha}Z^{-1}\pa^{-1})+\sum_{i=1}^\infty\pa^{-i}{\rm res}_\pa(\pa^{i-1}M^mZ\pa^{m+l}E_{\alpha}Z^{-1})\\
&=& {\rm res}_z\Big(M^mZ\pa^{m+l}E_{\alpha}(e^\xi)\cdot(Z^{-1}\pa^{-1})^*(e^{-\xi})\Big)+
\sum_{i=1}^\infty\pa^{-i}{\rm res}_z\Big(\pa^{i-1}M^mZ\pa^{m+l}E_{\alpha}(e^\xi)\cdot Z^{-1*}(e^{-\xi})\Big)\\
&=&{\rm res}_z\Big(z^{m+l}R_{\alpha}\pa_z^mw(t,z)w^*(t,z)\Big)-\sum_{i=1}^\infty \pa^{-i}{\rm res}_z \Big(z^{m+l}(R_{\alpha}\pa_z^mw(t,z))^{(i-1)}w^*(t,z)_x\Big)\\
&=&{\rm res}_z\Big(z^{m+l}R_{\alpha}\pa_z^m
w(t,z)\cdot\big(w^*(t,z)-\pa^{-1}w^*(t,z)_x\big)\Big)={\rm res}_z\Big(z^{m+l}R_{\alpha}\pa_z^m
w(t,z)\cdot\pa^{-1}w^*(t,z)\pa\Big),
\end{eqnarray*}
which further leads to
\begin{eqnarray*}
Y_{\alpha}(\lambda,\mu)&=&{\rm res}_z\left(\sum_{m=0}^\infty\sum_{l=-\infty}^\infty
\frac{z^{m+l}}{\lambda^{m+l+1}}\frac{1}{m!}(\mu-\lambda)^mR_{\alpha}\pa_z^mw(t,z)\cdot\pa^{-1}w^*(t,z)\pa\right)\\
&=&{\rm res}_z\left(\delta(\lambda,z)R_{\alpha}e^{(\mu-\lambda)\pa_z}
w(t,z)\cdot\pa^{-1}w^*(t,z)\pa\right)\\
&=&R_{\alpha}e^{(\mu-\lambda)\pa_\lambda} w(t,\lambda)\cdot\pa^{-1}w^*(t,\lambda)\pa
=(R_{\alpha}w(t,\mu))\cdot\pa^{-1}\cdot w^*(t,\lambda)\pa.
\end{eqnarray*}
The identity \eqref{geneaddsymmmkpwave} can be derived by using the fact
\[\sum_{\alpha=1}^NR_{\alpha}=E.\]
\end{proof}

Inspired by \cite{orlov3,orlov4,orlov5},  the ghost
symmetry flow for the multicomponent mKP hierarchy can be defined by the following equations
\begin{eqnarray}
\pa_\alpha L=\left[(R_{\alpha}w)\cdot\pa^{-1}\cdot w^*\pa,\ L\right],\quad \pa_\alpha P=(R_{\alpha}w)\cdot\pa^{-1}\cdot w^*\pa \cdot P.\label{paalphalz}
\end{eqnarray}
The actions of $\pa_\alpha$ on the wave function $\phi$ ( a linear combination of the eigenfunctions) and the adjoint wave function  $\phi^*$
( a linear combination of the adjoint eigenfunctions) are listed below.
\begin{eqnarray}
\pa_\alpha \phi=(R_{\alpha}w) \int\phi_xw^*dx,\quad \pa_\alpha \phi^*=-w^*\int(R_{\alpha}w)\phi^*_xdx.
\end{eqnarray}
The actions of $\pa_Y$ on the wave function  $\phi$ and the adjoint wave function  $\phi^*$ are listed below.
\begin{eqnarray}
\pa_Y \phi=w \int\phi_xw^*dx,\quad \pa_Y \phi^*=-w^*\int w\phi^*_xdx.
\end{eqnarray}

We can further get
The additional flows of $\partial_{t_{l,k,\alpha}}$ are  symmetry flows of the multicomponent mKP hierarchy, i.e. they commute with all $\partial_{t_{\beta, n}}$ flows of the  multicomponent mKP hierarchy.

According to the action of  $\partial_{t^*_{l,k,\alpha}}$ and $\partial_{t_{\beta,n}}$ on the
Lax operator $L$,  we can rewrite the quantum torus flow $\partial_{t^*_{l,k,\alpha}}$¡¡in terms of a combination of $\partial_{t_{p,s}}$ flows
\begin{eqnarray*}
\partial_{t^*_{l,k,\alpha}}L &=& [-(\sum_{p,s=0}^{\infty}\frac{l^p(k\log q)^sM^pL^sR_{\alpha}}{p!s!})_{\leq 0},L]\\
 &=&\sum_{p,s=0}^{\infty}\frac{l^p(k\log q)^s}{p!s!}\partial_{t_{p,s,\alpha}}L,
\end{eqnarray*}
which further leads to
\begin{eqnarray*}
[\partial_{t^*_{l,k,\alpha}},\partial_{t_{\beta,n}}]L &=& [\sum_{p,s=0}^{\infty}\frac{l^p(k\log q)^s}{p!s!}\partial_{t_{p,s,\alpha}},\partial_{t_{\beta,n}}]L\\
&=& \sum_{p,s=0}^{\infty}\frac{l^p(k\log q)^s}{p!s!}[\partial_{t_{p,s,\alpha}},\partial_{t_{\beta,n}}]L\\
 &=&0.
\end{eqnarray*}Then we can derive the following theorem.
\begin{theorem}
The additional flows $\partial_{t^*_{l,k,\alpha}}$ are  symmetries of the multicomponent mKP hierarchy, i.e. they commute with all $\partial_{t_{\beta,n}}$ flows of the  multicomponent mKP hierarchy.
\end{theorem}

Now it is time to identify the algebraic structure of the quantum torus
additional $\partial_{t_{l,k,\alpha}^*}$ flows of the multicomponent mKP hierarchy in the following theorem.
\begin{theorem}
The additional flows $\partial_{t_{l,k,\alpha}^*}$ of the multicomponent mKP hierarchy form the following multi quantum torus type algebra, i.e.,
\begin{equation}
[\partial_{t^*_{n,m,\beta}},\partial_{t^*_{l,k,\alpha}}]=(q^{ml}-q^{nk})\delta_{\alpha,\beta}\partial_{t^*_{n+l,m+k,\alpha}},\ \ n,m,l,k\geq 0, \ 1\leq \alpha,\beta\leq N.
\end{equation}

\end{theorem}
\begin{proof}
Using the Jacobi identity, we can derive the following computation which will finish the proof of this theorem

\begin{eqnarray*}
&&[\partial_{t^*_{n,m,\beta}},\partial_{t^*_{l,k,\alpha}}]L\\
&=&\partial_{t^*_{n,m}}([-(e^{lM}q^{kL}R_{\alpha})_{\leq 0},L])-\partial_{t^*_{l,k,\alpha}}([-(e^{nM}q^{mL}R_{\beta})_{\leq 0},L]) \\
&=&[-(\partial_{t^*_{n,m}} (e^{lM}q^{kL}R_{\alpha}))_{\leq 0},L] +[-(e^{lM}q^{kL}R_{\alpha})_{\leq 0},(\partial_{t^*_{n,m}} L)]
+[[-(e^{lM}q^{kL}R_{\alpha})_{\leq 0},e^{nM}q^{mL}R_{\beta}]_{\leq 0},L]\\
&&+[(e^{nM}q^{mL}R_{\beta})_{\leq 0},[-(e^{lM}q^{kL}R_{\alpha})_{\leq 0},L]] \\
&=&[[(e^{nM}q^{mL}R_{\beta})_{\leq 0},e^{lM}q^{kL}R_{\alpha}]_{\leq 0},L] +[(e^{lM}q^{kL}R_{\alpha})_{\leq 0},[(e^{nM}q^{mL}R_{\beta})_{\leq 0},L]]
\\
&&+[[-(e^{lM}q^{kL}R_{\alpha})_{\leq 0},e^{nM}q^{mL}R_{\beta}]_{\leq 0},L]+[(e^{nM}q^{mL}R_{\beta})_{\leq 0},[-(e^{lM}q^{kL}R_{\alpha})_{\leq 0},L]] \\
&=&[[(e^{nM}q^{mL}R_{\beta})_{\leq 0},e^{lM}q^{kL}R_{\alpha}]_{\leq 0},L] +[[(e^{lM}q^{kL}R_{\alpha})_{\leq 0},(e^{nM}q^{mL}R_{\beta})_{\leq 0}],L]
\\
&&+[[-(e^{lM}q^{kL}R_{\alpha})_{\leq 0},e^{nM}q^{mL}R_{\beta}]_{\leq 0},L]\\
&=&[[e^{nM}q^{mL}R_{\beta},e^{lM}q^{kL}R_{\alpha}]_{\leq 0},L]\\
&=&-(q^{ml}-q^{nk})[(\delta_{\alpha,\beta}e^{(n+l)M}q^{(m+k)L}R_{\alpha})_{\leq 0},L]\\
&=&(q^{ml}-q^{nk})\delta_{\alpha,\beta}\partial_{t^*_{n+l,m+k,\alpha}}L.
\end{eqnarray*}

One can also prove this theorem as following in another way by considering the functions' set $\{e^{mM}q^{nL}R_{\alpha},\ m,n\geq 0, \ 1\leq \alpha\leq N\}$ has an isomorphism with the functions' set $\{q^{nz}e^{m\partial_z}R_{\beta},\ m,n\geq 0, \ 1\leq \beta\leq N\}$ as
\begin{equation}
e^{mM}q^{nL}R_{\alpha} \qquad \mapsto\qquad  q^{nz}e^{m\partial_z}E_{\alpha},
\end{equation}
with the following commutator
\begin{equation}
[q^{nz}e^{m\partial_z}E_{\alpha},q^{lz}e^{k\partial_z}E_{\beta}]=\delta_{\alpha,\beta}(q^{ml}-q^{nk})q^{(n+l)z}e^{(m+k)\partial_z}.
\end{equation}

\end{proof}

\section{The constrained multi-component mKP hierarchy}

Similar as the constrained multicomponent KP hierarchy in \cite{zhangJPA}, the following reduction condition can also be imposed onto the $N$-component mKP hierarchy:
\[\label{constr}\sum_{l=1}^rd_l^s(LR_l)^s+\sum_{l=r+1}^NLR_l= \sum_{l=1}^rd_l^s(LR_l)_{\geq 1}^s+\sum_{l=r+1}^N(LR_l)_{\geq 1},\ \ 1\leq r\leq N.\]

Let us denote following operators as
\[\hat \d:=\d_1+\d_1+\dots+\d_r,\ \ \d_i=\frac{\d}{\d t_{i,1}},\]
\[\hat L^{(i)}=P^{(r)}(\hat \d)E_{i}^{(r)}\hat \d (P^{(r)}(\hat \d))^{-1},\ \
\hat L=\sum_{i=1}^rd_i\hat L^{(i)},\]
and
\[B_k^{(\beta,r)}=(P^{(r)}(\hat \d)E_\beta^{(r)}\hat \d^k (P^{(r)}(\hat \d))^{-1})_{\geq 1}:=(C_k^{(\beta,r)})_{\geq 1},\ \ 1\leq \beta\leq r,\]
\[C_k^{(\beta,r)}=P^{(r)}(\hat \d)E_\beta^{(r)}\hat \d^k (P^{(r)}(\hat \d))^{-1},\ \ 1\leq \beta\leq r,\]
\[\phi^{(r)}=(P_1)_{ij}|_{1\leq i\leq r,r+1\leq j\leq N},\ \ \psi^{(r)}=(P_1)_{ij}|_{1\leq j\leq r,r+1\leq i\leq N},\]
where $P^{(r)}(E_{i}^{(r)})$ is the $r\times r$ principal sub-matrix of $P(E_{i})$.
This further lead to
\[\L=\hat L^s=\sum_{l=1}^rd_l^sB_s^{(l,r)}+\phi^{(r)}\hat \d^{-1}\psi^{(r)}\hat \d.\]
Under the constraint eq.\eqref{constr}, the following evolutionary equations of the constrained $N$-component mKP hierarchy can be derived

\[\frac{\d \phi^{(r)}}{\d t_{n,\beta}}=B_n^{(\beta,r)}\phi^{(r)},\ \ \frac{\d \psi^{(r)}}{\d t_{n,\beta}}=-(B_n^{(\beta,r)})^*\psi^{(r)},\]
\[\frac{\d \hat L^s}{\d t_{n,\beta}}=[B_n^{(\beta,r)},\hat L^s],\]
where for $B_n^{(\beta,r)}=\sum_{i=1}^nB_i\hat \d^i,$
\[(B_n^{(\beta,r)})^*\psi^{(r)}:=\sum_{i=1}^n(-1)^i\hat \d^i(\psi^{(r)}B_i).\]
When $r=1,N=2$, one can obtain a $s$-constrained mKP hierarchy. When  $r=1,N\geq 2$, one can obtain the vector $s$-constrained mKP hierarchy.

\section{Virasoro symmetry of the  constrained multi-component mKP hierarchy}
In this section, we shall construct the additional symmetry and discuss the algebraic structure of the
additional symmetry  flows of the $N$-component  constrained mKP hierarchy.

To this end, firstly we define $\hat\Gamma^{(r)}$ and the Orlov-Shulman's  operator $\M$
 \begin{equation}
\hat\Gamma^{(r)}=\sum_{\beta=1}^rt_{1\beta}\frac{E_{\beta}}{sd_{\beta}^s}\hat \d^{1-s}+\sum_{\beta=1}^r\sum_{n=2}^{\infty}\frac nsd_{\beta}^{-s}E_{\beta}\hat \d^{n-s}t_{n,\beta},\ \ \M= P^{(r)} \hat\Gamma^{(r)} (P^{(r)})^{-1}.
\end{equation}

It  is easy to find the following formula
\begin{equation}\label{commute}
[\partial_{t_{n,\beta}}-\hat \d^nE_{\beta},\hat\Gamma^{(r)}]=0.
\end{equation}

There are the following commutation relations
\begin{equation}
[\sum_{\beta=1}^rd_{\beta}^s\hat \d^s E_{\beta},\hat\Gamma^{(r)}]=E,
\end{equation}
which can be verified by a straightforward calculation.
By using the Sato equation, the isospectral flow of the $\M$ operator
is given by
 \begin{equation}
\partial_{t_{n,\beta}}\M=[B_n^{(\beta,r)},\M].
\end{equation}  More generally,
\begin{equation}
\partial_{t_{n,\beta}}(\M^m\L^l)=[B_n^{(\beta,r)},\M^m\L^l].
\end{equation}

The Lax operator $\L$ and the Orlov-Shulman's $\M$ operator satisfy the following canonical relation
\[[\L,\M]=E.\] Then the additional flows for the time variable $t_{m,n,\beta}$ will be
defined as follows
\begin{equation}
\dfrac{\partial S}{\partial t_{m,n,\beta}}=-(\M^mC_n^{(\beta,r)})_{\leq 0}S,\ \ m,n \in \N, 1\leq \beta\leq r,
\end{equation}
which is equivalent to
\begin{equation}
\dfrac{\partial \L}{\partial t_{m,n,\beta}}=-[(\M^mC_n^{(\beta,r)})_{\leq 0},\L], \qquad
\dfrac{\partial \M}{\partial t_{m,n,\beta}}=-[(\M^mC_n^{(\beta,r)})_{\leq 0},\M].
\end{equation}
Later we can prove the additional flows $\dfrac{\partial }{\partial
t_{m,n,\beta}}$  commute with the flow $\dfrac{\partial
}{\partial t_{k,\gamma}}$ and they form a kind of $W_{\infty}$ infinite dimensional  additional Lie algebra which contain Virasoro algebra as a
subalgebra. To this purpose, we need several lemmas and propositions as preparation firstly.

For above local differential operators $B_n^{(\beta,r)}$, we have the following lemma.
\begin{lemma}
  \label{lemm} $[B_n^{(\beta,r)},\phi^{(r)}\hat \d^{-1}\psi^{(r)}\hat \d]_{\leq 0}=B_n^{(\beta,r)}(\phi^{(r)})\hat \d^{-1}\psi^{(r)}\hat \d-\phi^{(r)}\hat \d^{-1}(\hat \d^{-1}B_n^{(\beta,r)*}\hat \d)(\psi^{(r)})\hat \d$.
\end{lemma}
\begin{proof}
 Firstly we consider a fundamental monomial: $\hat \d^n$ ($n\ge 1$). Then
  \begin{displaymath}
    [\hat \d^n,\phi^{(r)}\hat \d^{-1}\psi^{(r)}\hat \d]_{\leq 0}=(\hat \d^n(\phi^{(r)}))\hat \d^{-1}\psi^{(r)}\hat \d- (\phi^{(r)}\hat \d^{-1}\psi^{(r)}\hat \d \hat \d^n)_{\leq 0}.
  \end{displaymath}
  Notice that the second term can be rewritten in the following way
  \begin{eqnarray*}
    (\phi^{(r)}\hat \d^{-1}\psi^{(r)}\hat \d \hat \d^{n})_{\leq 0}
    =(\phi^{(r)}\psi^{(r)} \hat \d^{n}
        -\phi^{(r)}\hat \d^{-1}(\hat \d \psi^{(r)})\hat \d^{n})_{\leq 0}=(\phi^{(r)}\hat \d^{-1}(-\hat \d \psi^{(r)})\hat \d^{n})_{\leq 0} \\
    =\cdots
    =\phi^{(r)}\hat \d^{-1}\left((-\hat \d)^n
    (\psi^{(r)})\right)\hat \d=\phi^{(r)}\hat \d^{-1}(\hat \d\hat \d^n\hat \d^{-1})^*(\psi^{(r)})\hat \d,
  \end{eqnarray*}
  then the lemma is proved.
\end{proof}

We can get some properties of the Lax operator in the following proposition.

\begin{lemma}The Lax operator $\L$ of the constrained $N$-component mKP hierarchy  will satisfy the relation of
\begin{equation}\label{Lk}
(\L^k)_{\leq 0}=\sum_{j=0}^{k-1}\L^{k-j-1}(\phi^{(r)})\hat \d^{-1}{(\hat \d\L^j\hat \d^{-1})}^*(\psi^{(r)})\hat \d,\ \ k\in \Z,
\end{equation}
where $\L(\phi):=(\L)_{\geq 0}(\phi) + \phi^{(r)}\hat \d^{-1}(\psi^{(r)}\hat \d\phi),$  for arbitrary $r\times r$ matrix function $\phi$.
\end{lemma}

The action of original
additional flows of the  constrained $N$-component mKP hierarchy  is  expressed by
\begin{equation}\label{additional}
(\partial_{1,k,\beta} \L)_{\leq 0}=[(\M C_k^{(\beta,r)})_{\geq 1},\L]_{\leq 0}+(C_k^{(\beta,r)})_{\leq 0}.
\end{equation}

To keep the consistency with flow equations on eigenfunction and adjoint eigenfunction $\phi^{(r)},\psi^{(r)}$, we shall introduce an operator
$F_k^{(\alpha,r)}$ as following to modify the additional symmetry
of the constrained multi-component mKP hierarchy.

We now introduce a pseudo  differential operator  $F_k^{(\alpha,r)}$,
\begin{eqnarray}
F_k^{(\alpha,r)}&=&0,k=-1,0,1,\label{y12}\\
F_k^{(\alpha,r)}&=&\sum_{j=0}^{k-1}[j-\frac{1}{2}(k-1)]C_{k-1-j}^{(\alpha,r)}(\phi^{(r)})
\hat \d^{-1}(\hat \d C_{j}^{(\alpha,r)}\hat \d^{-1})^*(\psi^{(r)})\hat \d,k\geq 2.\label{y}
\end{eqnarray}

The following
lemmas are
necessary to concern the Virasoro symmetry.
\begin{lemma}If $
X=M\hat \d^{-1}N\hat \d,$ then
\begin{equation}\label{operatorXL}
[X,\L]_{\leq 0}=[M\hat \d^{-1}(\hat \d\L\hat \d^{-1})^*(N)\hat \d-\L(M)\hat \d^{-1}N\hat \d]+ [X(\phi^{(r)})\hat \d^{-1}\psi^{(r)}\hat \d-\phi^{(r)}\hat \d^{-1}(\hat \d X\hat \d^{-1})^*(\psi^{(r)})\hat \d].
\end{equation}
\end{lemma}
\begin{lemma}
The action of flows $\partial_{t_{l,\beta}}$ of the  constrained $N$-component mKP hierarchy  on the
$F_k^{(\alpha,r)}$ is
\begin{equation}\label{ykderivat}
\partial_{t_{l,\beta}}
F_k^{(\alpha,r)}=[B_l^{(\beta,r)},F_k^{(\alpha,r)}]_{\leq 0}.
\end{equation}
\end{lemma}
\begin{proof}
\begin{eqnarray*}
\partial_{t_{l,\beta}}F_k^{(\alpha,r)}&=&\partial_{t_{l,\beta}}(\sum_{j=0}^{k-1}[j-\frac{1}{2}(k-1)]C_{k-1-j}^{(\alpha,r)}(\phi^{(r)})\hat \d^{-1}(\hat \d^{-1}C_{j}^{(\alpha,r)*}\hat \d)(\psi^{(r)})\hat \d)\\\nonumber
&=&\sum_{j=0}^{k-1}[j-\frac{1}{2}(k-1)]\{\partial_{t_{l,\beta}}(C_{k-1-j}^{(\alpha,r)}(\phi^{(r)}))\hat \d^{-1}(\hat \d^{-1}C_{j}^{(\alpha,r)*}\hat \d)(\psi^{(r)})\hat \d\\
&&+C_{k-1-j}^{(\alpha,r)}(\phi^{(r)})\hat \d^{-1}\partial_{t_{l,\beta}}((\hat \d^{-1}C_{j}^{(\alpha,r)*}\hat \d)(\psi^{(r)}))\hat \d\}\\\nonumber
&=&[B_l^{(\beta,r)}\circ \sum_{j=0}^{k-1}[j-\frac{1}{2}(k-1)]C_{k-1-j}^{(\alpha,r)}(\phi^{(r)})\hat \d^{-1}(\hat \d^{-1}C_{j}^{(\alpha,r)*}\hat \d)(\psi^{(r)})\hat \d]_{\leq 0}\\
&&-[(\sum_{j=0}^{k-1}[j-\frac{1}{2}(k-1)]C_{k-1-j}^{(\alpha,r)}(\phi^{(r)})\hat \d^{-1}(\hat \d^{-1}C_{j}^{(\alpha,r)*}\hat \d )(\psi^{(r)}))\circ B_l^{(\beta,r)}\hat \d]_{\leq 0}\\
&=&[B_l^{(\beta,r)},(\sum_{j=0}^{k-1}[j-\frac{1}{2}(k-1)]C_{k-1-j}^{(\alpha,r)}(\phi^{(r)})\hat \d^{-1}(\hat \d^{-1}C_{j}^{(\alpha,r)*}\hat \d)(\psi^{(r)}))\hat \d]_{\leq 0}\\
&=&[B_l^{(\beta,r)},F_k^{(\alpha,r)}]_{\leq 0}.
\end{eqnarray*}
\end{proof}

Further, the following expression of $[F_{k-1}^{(\beta,r)},\L]_{\leq 0}$ is also
necessary to define the additional flows of the  constrained $N$-component mKP hierarchy.
\begin{lemma}
The Lax operator $\L$ of the constrained $N$-component mKP hierarchy and $F_{k-1}^{(\beta,r)}$ has the following relation,
\begin{eqnarray}\notag
[F_{k-1}^{(\beta,r)},\L]_{\leq 0}&=&-(C_{k}^{(\beta,r)})_{\leq 0}+\frac{k}{2}[\phi^{(r)}\hat \d^{-1}(\hat \d C_{k-1}^{(\beta,r)}\hat \d^{-1})^*(\psi^{(r)})\hat \d+C_{k-1}^{(\beta,r)}(\phi^{(r)})\hat \d^{-1}\psi^{(r)}\hat \d\\
\label{ykl}
&&+F_{k-1}^{(\beta,r)}(\phi^{(r)})\hat \d^{-1}\psi^{(r)}\hat \d-\phi^{(r)}\hat \d^{-1}(\hat \d F_{k-1}^{(\beta,r)}\hat \d^{-1})^*(\psi^{(r)})\hat \d].
\end{eqnarray}
\end{lemma}
\begin{proof}
A direct calculation can lead to
\begin{eqnarray*}
[F_{k-1}^{(\beta,r)},\L]_{\leq 0}&=&[\sum^{k-2}_{j=0}[j-\frac{1}{2}(k-2)]C_{k-2-j}^{(\beta,r)}(\phi^{(r)})\hat \d^{-1}(\hat \d C_{j}^{(\beta,r)}\hat \d^{-1})^*(\psi^{(r)})\hat \d,\L]_{\leq 0}\\\nonumber
&=&\sum^{k-2}_{j=0}[j-\frac{1}{2}(k-2)]C_{k-2-j}^{(\beta,r)}(\phi^{(r)})\hat \d^{-1}(\hat \d C_{j+1}^{(\beta,r)}\hat \d^{-1})^*(\psi^{(r)})\hat \d\\\nonumber
&&-\sum^{k-2}_{j=0}[j-\frac{1}{2}(k-2)]C_{k-1-j}^{(\alpha,r)}(\phi^{(r)})\hat \d^{-1}(\hat \d C_{j}^{(\beta,r)}\hat \d^{-1})^*(\psi^{(r)})\hat \d\\\nonumber
&&+(F_{k-1}^{(\beta,r)}(\phi^{(r)})\hat \d^{-1}\psi^{(r)}\hat \d-\phi^{(r)}\hat \d^{-1}(\hat \d F_{k-1}^{(\beta,r)}\hat \d^{-1})^*(\psi^{(r)}))\hat \d\\\nonumber
&=&-\sum^{k-2}_{j=1}C_{k-1-j}^{(\beta,r)}(\phi^{(r)})\hat \d^{-1}(\hat \d C_{j}^{(\beta,r)}\hat \d^{-1})^*(\psi^{(r)})\hat \d\\\nonumber
&&+(\frac{k}{2}-1)[\phi^{(r)}\hat \d^{-1}(\hat \d C_{k-1}^{(\beta,r)}\hat \d^{-1})^*(\psi^{(r)})\hat \d+C_{k-1}^{(\beta,r)}(\phi^{(r)})\hat \d^{-1}\psi^{(r)}\hat \d]\\\nonumber
&&+(F_{k-1}^{(\beta,r)}(\phi^{(r)})\hat \d^{-1}\psi^{(r)}\hat \d-\phi^{(r)}\hat \d^{-1}(\hat \d F_{k-1}^{(\beta,r)}\hat \d^{-1})^*(\psi^{(r)}))\hat \d,
\end{eqnarray*}
which further help us deriving eq.\eqref{ykl} using eq.\eqref{Lk}.
\end{proof}

Putting together (\ref{additional}) and (\ref{ykl}), we define the
additional flows of the  constrained $N$-component mKP hierarchy  as
 \begin{equation}\label{tkflow}
\partial_{t_{1,k,\beta}}\L=[-(\M C_k^{(\beta,r)})_{\leq 0}+F_{k-1}^{(\beta,r)},\L],
\end{equation}
where
$F_{k-1}^{(\beta,r)}=0$, for $k=0,1,2$, such that the right-hand side of
(\ref{tkflow}) is in the form of derivation of Lax equations. Generally, one can also derive
\begin{equation}\label{MLK}
\partial_{t_{1,k,\beta}}(\M\L^l)=[-(\M C_k^{(\beta,r)})_{\leq 0}+F_{k-1}^{(\beta,r)},\M\L^l].
\end{equation}

Now we calculate the action of the additional flows eq.\eqref{tkflow}
on the eigenfunction $\phi^{(r)}$ and $\psi^{(r)}$ of the  constrained $N$-component mKP hierarchy.
\begin{theorem}\label{symmetre}
The acting of additional flows of  constrained $N$-component mKP hierarchy on the eigenfunction $\phi^{(r)}$ and $\psi^{(r)}$ are
\begin{equation}\label{BAfunction}
\begin{split}
{\partial_{t_{1,k,\beta}}\phi^{(r)}}&=(\M
 C_k^{(\beta,r)})_{\geq 1}(\phi^{(r)})+F_{k-1}^{(\beta,r)}(\phi^{(r)})+\frac{k}{2} C_{k-1}^{(\beta,r)}(\phi^{(r)}),\\
 {\partial_{t_{1,k,\beta}}\psi^{(r)}}&=-(\hat \d\M
 C_k^{(\beta,r)}\hat \d^{-1})^*_{\geq 1}(\psi^{(r)})-(\hat \d F_{k-1}^{(\beta,r)}\hat \d^{-1})^*(\psi^{(r)})+\frac{k}{2}{(\hat \d C_{k-1}^{(\beta,r)}\hat \d^{-1})^*} (\psi^{(r)}).
\end{split}
\end{equation}
\end{theorem}
\begin{proof}
Substitution of (\ref{ykl}) to negative part of (\ref{tkflow}) shows
\begin{eqnarray}\label{tkflow2}
\begin{split}
{(\partial_{t_{1,k,\beta}}\L)}_{\leq 0}&=(\M  C_k^{(\beta,r)})_{\geq 1}(\phi^{(r)})\hat \d^{-1}(\psi^{(r)})\hat \d-\phi^{(r)}\hat \d^{-1}(\hat \d \M  C_k^{(\beta,r)}\hat \d^{-1})_{\geq 1}^*(\psi^{(r)})\hat \d+F_{k-1}^{(\beta,r)}(\phi^{(r)})\hat \d^{-1}\psi^{(r)}\hat \d\\
&-\phi^{(r)}\hat \d^{-1}(\hat \d F_{k-1}^{(\beta,r)}\hat \d^{-1})^*(\psi^{(r)})\hat \d
+\frac{k}{2}\phi^{(r)}\hat \d^{-1}(\hat \d F_{k-1}^{(\beta,r)}\hat \d^{-1})^*(\psi^{(r)})\hat \d+\frac{k}{2}C_k^{(\beta,r)}(\phi^{(r)})\hat \d^{-1}\psi^{(r)}\hat \d.
\end{split}
\end{eqnarray}

On the other side,
\begin{equation}\label{tkl}
{(\partial_{t_{1,k,\beta}}\L)}_{\leq 0}=(\partial_{t_{1,k,\beta}}\phi^{(r)})\hat \d^{-1}\psi^{(r)}\hat \d+\phi^{(r)}\hat \d^{-1}{(\partial_{t_{1,k,\beta}}\psi^{(r)})}\hat \d.
\end{equation}
Comparing right hand sides of (\ref{tkflow2}) and (\ref{tkl})
implies  the action of additional flows on the eigenfunction and the
adjoint eigenfunction (\ref{BAfunction}).
\end{proof}

Next according the action of  $\partial_{t_{1,k,\beta}}$ and $\partial_{t_{l,\alpha}}$ on the
dressing operator $S$, then
\begin{eqnarray*}
[\partial_{t_{1,k,\beta}},\partial_{t_{l,\alpha}}]S &=& -\partial_{t_{1,k,\beta}}((C_{l}^{(\alpha,r)})_{\leq 0}
S)-\partial_{t_{l,\alpha}}[-(\M C_{k}^{(\beta,r)})_{\leq 0}+F_{k-1}^{(\beta,r)}]S \\
&=&(-\partial_{t_{1,k,\beta}}C_{l}^{(\alpha,r)})_{\leq 0}S-(C_{l}^{(\beta,r)})_{\leq 0}\partial_{t_{1,k,\beta}}S-[(\M
C_{k}^{(\beta,r)})_{\leq 0}-F_{k-1}^{(\beta,r)}](C_{l}^{(\alpha,r)})_{\leq 0}S \\
 &&+[(C_{l}^{(\alpha,r)})_{\geq 1},\M C_{k}^{(\beta,r)}]_{\leq 0}S-(\partial_{t_{l,\alpha}} F_{k-1}^{(\beta,r)})S\\
&=&[(C_{l}^{(\alpha,r)})_{\leq 0},-F_{k-1}^{(\beta,r)}]_{\leq 0}S+[-F_{k-1}^{(\beta,r)},C_{l}^{(\alpha,r)}]_{\leq 0}S-(\partial_{t_{l,\alpha}} F_{k-1}^{(\beta,r)})S \\
&=&[B_{l}^{(\alpha,r)},F_{k-1}^{(\beta,r)}]_{\leq 0}S-(\partial_{t_{1,\alpha}} F_{k-1}^{(\beta,r)})S\\
&=&0.
\end{eqnarray*}
Therefore the additional flows of $\partial_{t_{1,k,\beta}}$ are  symmetry flows of the  constrained $N$-component mKP hierarchy, i.e. they commute with all $\partial_{t_{l,\alpha}}$ flows of the   constrained $N$-component mKP hierarchy.

Taking into account $F_{k-1}^{(\beta,r)}=0$  for $k=0,1,2 $, then
eq.\eqref{BAfunction} becomes
\begin{eqnarray}\label{PLqr}
\begin{split}
\partial_{t_{1,l,\beta}}\phi^{(r)}&=(\M
C_{l}^{(\beta,r)})_{\geq 1}(\phi^{(r)})+\frac{1}{2}l C_{l-1}^{(\beta,r)} \phi^{(r)},\ \ l=0,1,2,\\
\partial_{t_{1,l,\beta}}\psi^{(r)}&=-(\hat \d\M
C_{l}^{(\beta,r)}\hat \d^{-1})^*_{\geq 1}(\psi^{(r)})+\frac{1}{2}l (\hat \d C_{l-1}^{(\beta,r)}\hat \d^{-1})^* \psi^{(r)}, \ \ l=0,1,2.
\end{split}
\end{eqnarray}

Then  using eq.\eqref{PLqr} and the
relation
${{\partial_{t_{1,l,\beta}}}(\L^k(\phi^{(r)}))}=({\partial_{t_{1,l,\beta}}}(\L^k))(\phi^{(r)})+\L^k
{\partial_{t_{1,l,\beta}}}(\phi^{(r)})$,
we can find
the additional  flows $\partial_{t_{1,l,\beta}}$ of  constrained $N$-component mKP hierarchy have the following relations
\begin{eqnarray}\label{lqstar}
\begin{split}
{\partial_{t_{1,l,\beta}} C_{k}^{(\alpha,r)}(\phi^{(r)})}&=(\M
C_{l}^{(\beta,r)})_{\geq 1}(C_{k}^{(\alpha,r)}(\phi^{(r)}))+(k+\frac{l}{2})C_{k+l-1}^{(\alpha,r)}\delta_{\alpha,\beta}(\phi^{(r)})+T_{l-1}^{(\beta,r)}C_{k}^{(\alpha,r)}(\phi^{(r)}),\\
{\partial_{t_{1,l,\beta}} (\hat \d C_{k}^{(\alpha,r)}\hat \d^{-1})^*(\psi^{(r)})}&=-(\hat \d \M
C_{l}^{(\beta,r)}\hat \d^{-1})^*_{\geq 1}(\hat \d C_{k}^{(\alpha,r)}\hat \d^{-1})^*(\psi^{(r)})+(k+\frac{l}{2})(\hat \d C_{k+l-1}^{(\alpha,r)}\hat \d^{-1})^*\delta_{\alpha,\beta}(\psi^{(r)})\\
&+(\hat \d T_{l-1}^{(\beta,r)}C_{k}^{(\alpha,r)}\hat \d^{-1})^*(\phi^{(r)}).
\end{split}
\end{eqnarray}

Moreover, the action of
 $\partial_{t_{1,l,\beta}}$ on $F_k^{(\alpha,r)}$ is given by the following lemma.
\begin{lemma}
The actions on $F_k^{(\alpha,r)}$ of the additional  symmetry flows
$\partial_{t_{1,l,\beta}}$ of the  constrained $N$-component mKP hierarchy are
\begin{eqnarray}\label{PL}
\partial_{t_{1,l,\beta}} F_k^{(\alpha,r)}=[(\M C_l^{\beta,r})_{\geq 1}+F_{l-1}^{(\beta,r)},F_k^{(\alpha,r)}]_{\leq 0}+(k-l+1)F_{k+l-1}^{(\alpha,r)}\delta_{\alpha,\beta}.
\end{eqnarray}
\end{lemma}
\begin{proof} Using eq.\eqref{lqstar}, a straightforward calculation implies
\begin{eqnarray}\label{partialyk}
\begin{split}
\partial_{t_{1,l,\beta}} F_k^{(\alpha,r)}&=\partial_{t_{1,l,\beta}} (\sum_{j=0}^{k-1}[j-\frac{1}{2}(k-1)]C_{k-1-j}^{(\alpha,r)}(\phi^{(r)})\hat \d^{-1}(\hat \d^{-1}C_{j}^{(\alpha,r)*}\hat \d)(\psi^{(r)}))\\
&=\sum_{j=0}^{k-1}[j-\frac{1}{2}(k-1)](\partial_{t_{1,l,\beta}}(C_{k-1-j}^{(\alpha,r)}(\phi^{(r)}))\hat \d^{-1}(\hat \d^{-1}C_{j}^{(\alpha,r)*}\hat \d)(\psi^{(r)})\\
&+C_{k-1-j}^{(\alpha,r)}(\phi^{(r)})\hat \d^{-1}(\partial_{t_{1,l,\beta}}(\hat \d^{-1}C_{j}^{(\alpha,r)*}\hat \d)(\psi^{(r)})))\\
&=\sum_{j=0}^{k-1}[j-\frac{1}{2}(k-1)][(\M C_l^{\beta,r})_{\geq 1}+F_{l-1}^{(\beta,r)}](C_{k-1-j}^{(\alpha,r)}(\phi^{(r)}))\hat \d^{-1}(\hat \d C_l^{\beta,r}\hat \d^{-1})^*(\psi^{(r)})\\
&+\sum_{j=0}^{k-1}[j-\frac{1}{2}(k-1)](k-j-1+\frac{l}{2})C_{k+l-2-j}^{\alpha,r}(\phi^{(r)})\hat \d^{-1}(\hat \d C_j^{\beta,r}\hat \d^{-1})^*(\psi^{(r)})\delta_{\alpha,\beta}\\
&-\sum_{j=0}^{k-1}[j-\frac{1}{2}(k-1)]C_{k-1-j}^{(\alpha,r)}(\phi^{(r)})\hat \d^{-1}[(\hat \d \M C_l^{\beta,r}\hat \d^{-1})^*_{\geq 1}+(\hat \d F_{l-1}^{(\beta,r)}\hat \d^{-1})^*](\hat \d C_j^{\beta,r}\hat \d^{-1})^*(\psi^{(r)})\\
&+\sum_{j=0}^{k-1}[j-\frac{1}{2}(k-1)](j+\frac{l}{2})C_{k-1-j}^{(\alpha,r)}(\phi^{(r)})\hat \d^{-1}(\hat \d C_{j+l-1}^{\beta,r}\hat \d^{-1})^*(\psi^{(r)})\delta_{\alpha,\beta},
\end{split}
\end{eqnarray}
which can be simplified  to eq.\eqref{PL}.
\end{proof}

Now it is time to identity the algebraic structure of the
additional symmetry flows of the  constrained $N$-component mKP hierarchy.
\begin{theorem}\label{alg}
The additional flows $\partial_{t_{1,k,\beta}}$ of the  constrained $N$-component mKP hierarchy form the
positive half of Virasoro algebra, i.e., for $l, k\geq 0,1\leq \alpha,\beta\leq r,$
\begin{equation}
[\partial_{t_{1,l,\alpha}},\partial_{t_{1,k,\beta}}]=(k-l)\delta_{\alpha,\beta}\partial_{t_{1,k+l-1,\alpha}}.
\end{equation}
\end{theorem}
\begin{proof}
Here we only consider the case when $l=0,1,2.$ Using $[(\M C_{l}^{(\alpha,r)})_{\leq 0},F_{k-1}^{(\beta,r)}]_{\leq 0}=[(\M C_{l}^{(\alpha,r)})_{\leq 0},F_{k-1}^{(\beta,r)}]$
and the Jacobi identity, one can derive the following computation

\begin{eqnarray*}
&&[\partial_{t_{1,l,\alpha}},\partial_{t_{1,k,\beta}}]\L\\
&=&\partial_{t_{1,l,\alpha}}([-(\M C_{k}^{(\beta,r)})_{\leq 0},\L]+[F_{k-1}^{\beta,r},\L])-\partial_{t_{1,k,\beta}}([-(\M
 C_l^{\alpha,r})_{\leq 0},\L]+[F_{l-1}^{\alpha,r},\L])) \\
&=&\partial_{t_{1,l,\alpha}} [-(\M C_{k}^{(\beta,r)})_{\leq 0},\L] +[\partial_{t_{1,l,\alpha}} F_{k-1}^{\beta,r},\L]+[F_{k-1}^{\beta,r},\partial_{t_{1,l,\alpha}} \L]+[\partial_{t_{1,k,\beta}}(\M
 C_l^{\alpha,r})_{\leq 0},\L]\\
&&+[(\M
 C_l^{\alpha,r})_{\leq 0},\partial_{t_{1,k,\beta}} \L]-[\partial_{t_{1,k,\beta}} F_{l-1}^{\alpha,r},\L]-[F_{l-1}^{\alpha,r},\partial_{t_{1,k,\beta}} \L]\\
&=&[[(\M C_{l}^{(\alpha,r)})_{\leq 0},\M C_k^{\beta,r}]_{\leq 0},\L] +[-(\M C_{k}^{(\beta,r)})_{\leq 0},[-(\M C_l^{(\alpha,r)})_{\leq 0}
,\L]] \\
&&+[[(\M C_l^{\alpha,r})_{\geq 1},F_{k-1}^{\beta,r}]_{\leq 0}+(k-l)F_{k+l-2}^{(\beta,r)}\delta_{\alpha,\beta},\L]\\
&&+[F_{k-1}^{\beta,r},[-(\M C_l^{\alpha,r})_{\leq 0},\L]]
+[[-(\M C_{k}^{(\beta,r)})_{\leq 0}+F_{k-1}^{\beta,r},\M C_l^{\alpha,r}]_{\leq 0},\L]\\
&&+[(\M
C_l^{\alpha,r})_{\leq 0},[-(\M C_{k}^{(\beta,r)})_{\leq 0}+F_{k-1}^{\beta,r},\L]]\\
&=&[[(\M C_{l}^{(\alpha,r)})_{\leq 0},\M C_k^{\beta,r}]_{\leq 0},\L] +[-(\M C_{k}^{(\beta,r)})_{\leq 0},[-(\M C_l^{(\alpha,r)})_{\leq 0}
,\L]] \\
&&+[(k-l)F_{k+l-2}^{(\beta,r)}\delta_{\alpha,\beta},\L]+[[-(\M C_{k}^{(\beta,r)})_{\leq 0},\M C_l^{\alpha,r}]_{\leq 0},\L]\\
&&+[(\M
C_l^{\alpha,r})_{\leq 0},[-(\M C_{k}^{(\beta,r)})_{\leq 0},\L]]\\
&=&(k-l)[-\delta_{\alpha,\beta}(\M C_{l+k-1}^{(\alpha,r)})_{\leq 0},\L] +[(k-l)F_{k+l-2}^{(\beta,r)}\delta_{\alpha,\beta},\L]\\
&=&(k-l)\delta_{\alpha,\beta}\partial_{t_{1,k+l-1,\alpha}}\L.
\end{eqnarray*}
The other case can be similarly proved.
\end{proof}

{\bf Acknowledgments:}\\
{\noindent \small
Chuanzhong Li  is  supported by the National Natural Science Foundation of China under Grant No. 11571192 and K. C. Wong Magna Fund in
Ningbo University.
}


\end{document}